\documentclass[10pt,twocolumn,twoside]{IEEEtran}

\usepackage{epsf,color}
\usepackage{graphicx}
\usepackage{amsmath,bm}
\usepackage{amssymb}
\usepackage{epsfig,latexsym,amsmath,epsf,amssymb,amsfonts}
\usepackage{verbatim}
\usepackage{placeins}
\usepackage[hang]{subfigure}
\usepackage{epstopdf}
\usepackage{amsthm}
\usepackage{url}
\usepackage{algorithm}
\usepackage{algorithmic}
\usepackage{caption}
\usepackage{cite}

\usepackage{soul,xcolor}
\setstcolor{red}

\newtheorem{theorem}{Theorem}

\newcommand{\prob}{\mathbb{P}}

\begin{document}
\pagestyle{empty}
\title{Multi-user Beam-Alignment\\for Millimeter-Wave Networks}

\author{Rana A. Hassan, Nicolo Michelusi
\thanks{This research has been funded by NSF under grant CNS-1642982.}
\thanks{R. A. Hassan and N. Michelusi are with the School of Electrical and Computer Engineering, Purdue University. email: \{hassan45,michelus\}@purdue.edu.}
\vspace{-12mm}
}


\maketitle

\begin{abstract}
Millimeter-wave communications 
is the most promising technology for next-generation cellular wireless systems,
thanks to the large bandwidth available compared to sub-6 GHz networks.
 Nevertheless, communication at these frequencies requires narrow beams via massive MIMO and beamforming to overcome the strong signal attenuation, and thus
 precise beam-alignment between transmitter and receiver is needed.
  The resulting signaling overhead may become a severe impairment, especially in mobile networks with high users density.
  Therefore, it is imperative to optimize the beam-alignment protocol to minimize the signaling overhead.
In this paper, the design of energy efficient joint beam-alignment protocols for two users
is addressed, with the goal to minimize the power consumption during data transmission, subject to rate constraints for both users, under analog beamforming constraints. It is proved that a bisection search algorithm is optimal. Additionally,
 the optimal scheduling strategy of the two users in the data communication phase is optimized based on the outcome of beam-alignment, according to a time division  multiplexing scheme.
The numerical results show significant decrease in the power consumption for the proposed joint beam-alignment scheme compared to exhaustive search and a single-user beam-alignment scheme taking place separately for each user.
\end{abstract}
\vspace{-5mm}
\section{Introduction}
Mobile data traffic has shown a tremendous growth in the past few decades, and is expected to increase by $53\%$ in each year until 2020 \cite{cisco2015cisco}. 
Traditionally, mobile data traffic is served almost exclusively by wireless systems operating under $6$ GHz, due to the availability of low-cost hardware and favorable propagation characteristics at these frequencies.
However,
conventional sub-$6$ GHz networks cannot support the high data rate
required by applications such as high definition video streaming,
 due to limited bandwidth availability.
 For this reason, millimeter-wave (mm-wave) systems operating between $30$ to $300$ GHz are receiving growing interest in both 5G related research and industry \cite{niu2015survey},\cite{rappaport2014millimeter}.


The large bandwidth available in the mm-wave frequency band can better address the demands of the ever increasing mobile traffic.
However, signal propagation at these frequencies is more challenging than traditional sub-6 GHz systems,
due to factors such as high propagation loss, directivity, sensitivity to blockage \cite{rappaport1996wireless}, which are exacerbated with the increase in the carrier frequency. These  features
 open up many challenges in both physical and MAC layers for the mm-wave frequencies to support high data rate.
To overcome the propagation loss, mm-wave systems
are expected to leverage
narrow beam communication, via large-dimensional antenna arrays with directional beamforming at both base stations (BSs) and mobile users (MUs), as well as signal processing techniques such as precoding and combining~\cite{akdeniz2014millimeter}.

Maintaining beam-alignment between transmitter and receiver is a challenging task in mm-wave networks, especially in dense and mobile networks: 
under high user density and mobility, frequent blockages and loss of alignment may occur, requiring frequent realignment.
 Unfortunately, the beam-alignment protocol may consume time, frequency and energy resources, thus potentially offsetting the benefits of mm-wave directionality. Motivated by this fact, in our previous work \cite{muddassar} we derived the optimal beam-width for communication, number of sweeping beams, and transmission energy so as to maximize the average rate under an average power constraint in a mobile scenario with a single user.
Several schemes have been proposed to achieve beam-alignment in mm-wave networks.
One of the most popular ones is exhaustive search, where the BS and the MU sequentially search through all possible combinations of transmit and receive beam patterns \cite{jeong2015random}. An iterative search algorithm is proposed in \cite{hur2013millimeter}, where the BS first searches in wider sectors by using
wider beams, and then refines the search within the best sector. In \cite{hussain2017throughput}, we derived a throughput-optimal search scheme called bisection search, which refines search within the previous best sector by using a beam with half the width of the previous best sector. It is shown that the bisection scheme outperforms both iterative and exhaustive schemes in terms of maximizing throughput in the communication phase. All these works focus on a single-user scenario and do not investigate how to exploit the beam-alignment protocols jointly across multiple users.

In the literature, multiuser mm-wave systems have been studied under the topic of precoding \cite{alkhateeb2015limited}, \cite{alkhateeb2015achievable}, beamforming \cite{stirling2015multi} and for wideband mm-wave systems, where the channel is characterized by multi-path components, different delays,
Angle-of-Arrivals/Angle-of-Departures (AoAs/AoDs), and Doppler
shifts \cite{song2017robust}. In all of the previous work, the authors proposed new algorithms in order to enhance the system performance. However, the optimality with respect to optimizing the communication performance in multi-user settings is not established. All of these algorithms cost
in terms of time and energy resources, and have a large effect on
the directionality achieved in the data communication phase, and thus
on power consumption and achievable rate. This motivates us to seek how to optimally balance resources among beam-alignment and data communication.

In this paper, we consider the optimization of beam-alignment and data communication in a two-users mm-wave network. The BS transmits a sequence of beam-alignment beacons 
using a sequence of beams with different beam-shape, and refines its estimate on the 
position of the two users based on the feedback received. Afterwards,
it schedules data transmission to the two users via time-division.
Using a Markov decision process (MDP) formulation \cite{bertsekas1995dynamic}, we prove the optimality of a bisection search scheme during beam-alignment, which scans half 
of the uncertainty region associated to each user in each beam-alignment slot.
We demonstrate numerically power savings up to $3\mathrm{dB}$ lower than under exhaustive search.

The rest of the paper is organized as follows. In Section \ref{sysmo}, we present the system model and the problem formulation, followed by the analysis in Section \ref{analysis}.
Numerical results are presented in Section \ref{numres}, followed by concluding remarks in
Section \ref{conclu}. The proofs of the main analytical results are provided in the Appendix.
\vspace{-2mm}
\section{System Model}
\label{sysmo}

  We consider a mm-wave cellular network with a single base station (BS) and $M$ mobile users (MU$_{i}$), where $i{=}1,2,\cdots,M $, depicted in Fig. \ref{system_model}. 
  In this paper, we consider the case $M{=}2$, and leave the more general case $M{\geq}2$ for future work.

  The BS is located at the origin and the mobile user MU$_i$ is located at angular coordinate $\Theta_i$, at distance $d_i$ from the BS, where $\Theta_i \in [0,2\pi]$ and $d_i\leq d_{\max}$, with $d_{\max}>0$ being the coverage area of the BS. We assume that $\Theta_i$ is uniformly distributed in $[\frac{-\sigma}{2},\frac{\sigma}{2}]$, denoted as $\Theta_i \sim \mathcal{U}[\frac{-\sigma}{2},\frac{\sigma}{2}]$ where $\sigma \in (0,2\pi]$ represents the availability of prior information on the angular coordinate of MU$_{i}$. 
  We assume a single signal path between the BS and each MU,
  either line-of-sight (LOS) or a strong non-LOS signal (\emph{e.g.}, when the LOS signal is temporarily obstructed due to mobility).
  
  \begin{figure}
\begin{center}
\includegraphics[width=0.6\columnwidth]{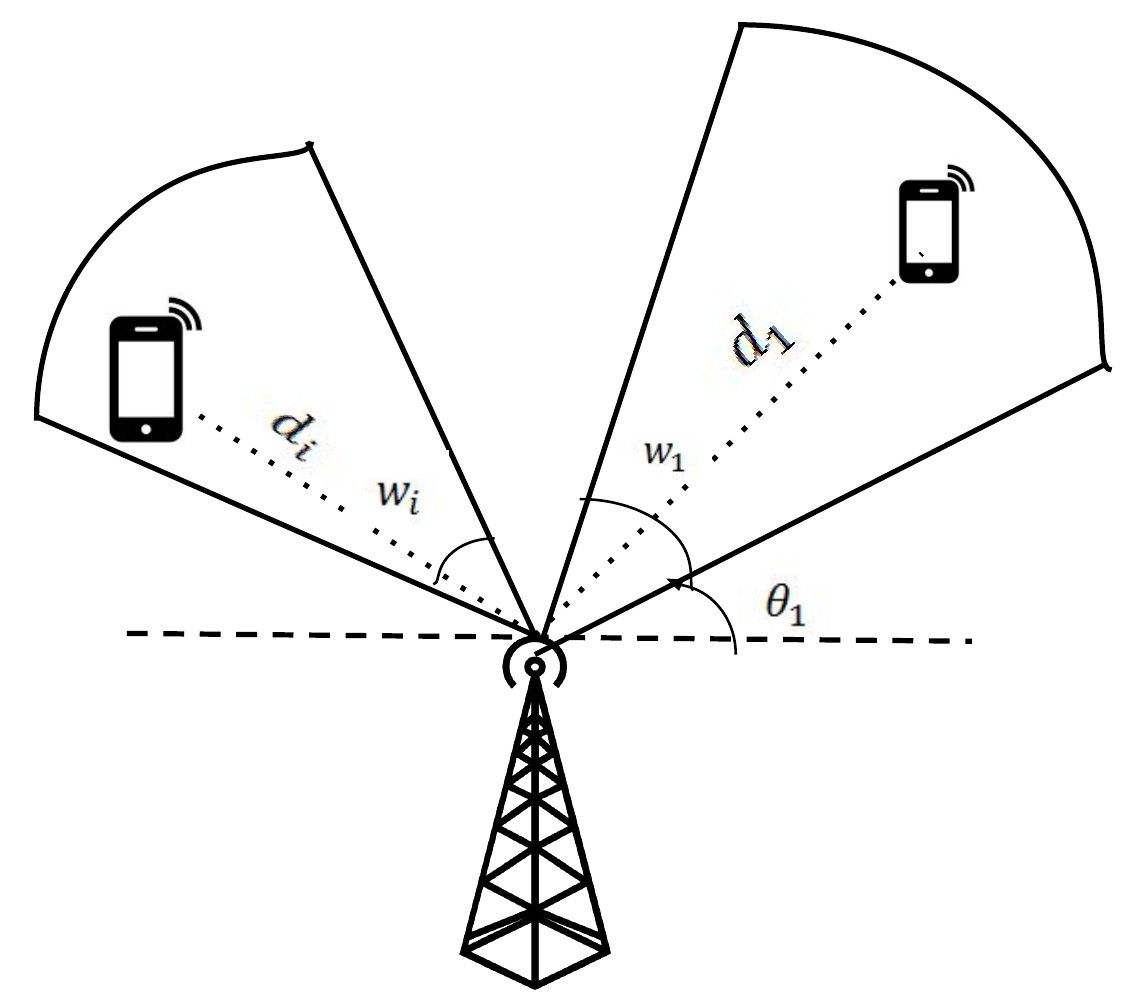}
\caption{Beam pattern for multiuser system model.}
\label{system_model}
\end{center}
\vspace{-5mm}
\end{figure}

The BS uses analog beamforming with a single RF chain. Data transmission is orthogonalized across users.
We model the transmission beam of the BS using a generalization of the \textit{sectored antenna model} \cite{bai2015coverage}: the overall transmission beam is the superposition of multiple beams, each covering a specific sector, which can be implemented via phase shifters \cite{abari2016millimeter}.
In addition, we ignore the effect of secondary beam lobes. 
Thus, we represent the beam shape (part of our design) at time $k$ via the set $\mathcal B_k\subseteq [-\pi,\pi]$, 
which represents the set of angular directions covered by the transmission beam.
Furthermore, we assume that the MUs receive isotropically.
 The proposed analysis can be extended to non-isotropic MUs by using multiple beam-alignment stages, each corresponding to a specific beam pattern at  the MU \cite{exhaustive}.
   
  We assume a frame-slotted network with frame duration $T_{\mathrm{fr}}$ [s]. Each frame is divided into a beam-alignment phase of duration $T_{\mathrm{BA}}$ (Sec. \ref{BA}), followed by a data communication phase of duration $T_{\mathrm{cm}}{=}T_{\mathrm{fr}}{-}T_{\mathrm{BA}}$ (Sec. \ref{comm}), shown in Fig.~\ref{fig.1}.
    \vspace{-4mm}
\subsection{Beam-Alignment Phase}
\label{BA}
In this section, we describe the beam-alignment phase, executed in the initial portion of the frame, of duration $T_{\mathrm{BA}}$.
Beam-alignment is performed over $L$ slots, each of duration $T\triangleq T_{\mathrm{BA}}/L$.
  As shown in Fig \ref{fig.1}, at the beginning of each slot $k=0,1,\dots, L-1$, the BS sends a beacon $b_{k}$ of duration $T_b<T$, using a beam with beam-shape $\mathcal B_k$, and receives a feedback message from both MUs in the remaining portion $T-T_b$ of the slot.

\begin{figure}
\begin{center}
\includegraphics[width=0.9\columnwidth]{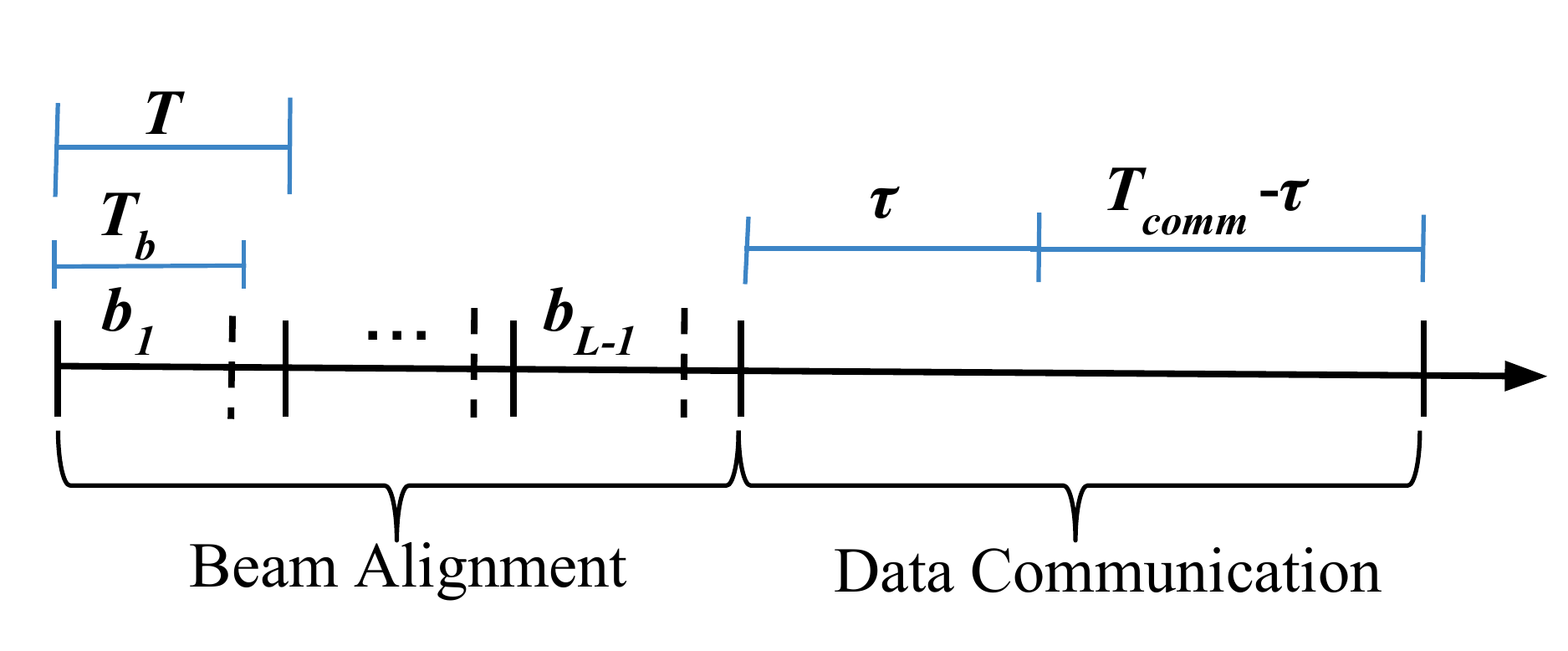}
\caption{Timing diagram of the beam-alignment and data communication phases.}
\label{fig.1}
\end{center}
\vspace{-5mm}
\end{figure}
  
  The beam-shape $\mathcal B_k$ is designed based on the current  
  probability density function (PDF) of MUs' angles $(\Theta_1,\Theta_2)$, denoted as $S_{k}(\theta_1,\theta_2)$, which is updated via Bayes' rule
    based on the feedback received from both MUs, see \eqref{Bayes}.
We also let 
  $S_{k,1}(\theta_1)=\int_{-\pi}^{\pi}S_{k}(\theta_1,\theta_2)\mathrm d\theta_2$
  and
  $S_{k,2}(\theta_2)=\int_{-\pi}^{\pi}S_{k}(\theta_1,\theta_2)\mathrm d\theta_1$
  be the marginal PDF of MU$_1$ and MU$_2$, respectively.
  Note that, at time $0$, $\Theta_i \sim \mathcal{U}[\frac{-\sigma}{2},\frac{\sigma}{2}]$, hence
  \begin{align}
  \label{init}
&  S_{0,1}(\theta)=S_{0,2}(\theta)=\frac{1}{\sigma}\chi\left(\theta\in\left[\frac{-\sigma}{2},\frac{\sigma}{2}\right]\right),
  \\
&  S_{0}(\theta_1,\theta_2)=S_{0,1}(\theta_1)\cdot S_{0,2}(\theta_2),
  \end{align}
  where $\chi(.)$ is the indicator function.
  For convenience, we define the support of $S_{k,i}$ as
  \begin{align}
  \mathcal S_{k,i}\triangleq\mathrm{supp}(S_{k,i}),
  \end{align}
  which defines the \emph{region of uncertainty} for MU$_i$ at time $k$.
  
  If MU${_{i}}$ is located within $\mathcal B_k$, \emph{i.e.}, $\Theta_i\in\mathcal B_k$,
  then it detects the beacon signal successfully and it transmits an acknowledgment (ACK) 
  back to the BS, denoted as $c_{i,k}{=}1$. Otherwise, it sends a negative-ACK (NACK), denoted as $c_{i,k}{=}0$, to inform the BS that no beacon has been detected.
  We assume that 
the feedback message $c_{i,k}\in\{0,1\}$ is received without error by the BS, within the end of the slot. This can be accomplished over a reliable low-frequency control channel, which does not require directional transmission and reception \cite{rangan2014millimeter}.
Additionally, we assume that the beacon is detected with no false-alarm nor mis-detection errors. This assumption requires a dedicated beam design to achieve small error probabilities \cite{hussain2017neyman}.
Thus, we can express the feedback signal as
\begin{align}
\label{feedback}
C_{k,i}=\chi(\Theta_i\in\mathcal B_k).
\end{align}

Given the sequence of feedback signals $C^{k}\triangleq(c_{0,1},c_{0,2}\cdots c_{k,1},c_{k,2})$
received up to slot $k$, and the sequence of beam shapes $\mathcal B^{k}\triangleq(\mathcal B_{0},\cdots,\mathcal B_{k})$ used for beam-alignment, the BS updates the PDF on the MUs' angular coordinate based on Bayes' rule as
\begin{align}
\label{Bayes}
&S_{k+1}(\theta_1,\theta_2)=f(\theta_1, \theta_2\mid \mathcal B^k,C^{k-1},c_{k,1},c_{k,2}) \\&\!\!\overset{(a)}{=}
\frac{\prob(c_{k,1},c_{k,2}\mid\theta_1, \theta_2,\mathcal B^k,C^{k-1})f(\theta_1, \theta_2 \mid \mathcal B^k,C^{k-1})}{\int_{[-\pi,\pi]^2}\prob(c_{k,1},c_{k,2}\mid\tilde{\boldsymbol{\theta}},\mathcal B^k,C^{k-1})f(\tilde{\boldsymbol{\theta}}\mid \mathcal B^k,C^{k-1})\mathrm d\tilde{\boldsymbol{\theta}}} \nonumber\\&\!\!\overset{(b)}{=}
\frac{\prob(c_{k,1}\mid\theta_1,\mathcal B_k)\prob(c_{k,2}\mid\theta_2,\mathcal B_k)
S_{k}(\theta_1,\theta_2)}{\int _{-\pi}^{\pi}\int_{-\pi}^{\pi}
\prob(c_{k,1}\mid\tilde\theta_1,\mathcal B_k)\prob(c_{k,2}\mid\tilde\theta_2,\mathcal B_k)
S_{k}(\tilde\theta_1,\tilde\theta_2)\mathrm d\tilde\theta_{1}\mathrm d\tilde\theta_{2}},
\nonumber
\end{align}
where $f(\cdot|\cdot)$ denotes conditional PDF. In step $(a)$, we applied Bayes' rule; in step $(b)$, we used the previous PDF $f(\theta_1, \theta_2|\mathcal B^k,C^{k-1}_{i})=S_k(\theta_1,\theta_2)$ and the fact that 
$c_{k,i}$ is a function of $\Theta_i$ and $\mathcal B_k$ via \eqref{feedback}.
\vspace{-2mm}
\subsection{Communication Phase}
\label{comm}
  In the communication phase of duration $T_{\mathrm{cm}}$, the BS
 schedules the two MUs using time division multiplexing (TDM).
 Specifically, it transmits to MU$_1$ over a portion $\tau_1\leq T_{\mathrm{cm}}$ of the data communication interval, using the transmission power $P_{L,1}$ and beam with shape $\mathcal B_{L,1}$, and to MU$_2$ over the remaining interval of duration $T_{\mathrm{cm}}-\tau_1$, with power $P_{L,2}$ and using a beam with shape  $\mathcal B_{L,2}$.
 
The powers $P_{L,i}$ and beam-shapes $\mathcal B_{L,i}$ for both MUs, and the time allocation $\tau_1$
are designed based on the PDF of the MUs' angular direction $S_L$ at the beginning of the communication phase,
so as to support the rate $R_i$ over the entire frame.
The beam-shape $\mathcal B_{L,i}$ for MU$_i$ is chosen so
as to provide coverage to guarantee successful transmission, \emph{i.e.}, 
\begin{align}
\label{BS}
\mathcal B_{L,i}=\mathcal S_{L,i}.
\end{align}
Thus, we can express the rate $R_i$ [bps/Hz] for both MUs as
 \begin{align}
& R_1=\frac{\tau_1}{T_{\mathrm{fr}}}\log_{2}\left(1+\gamma_1\frac{P_{L,1}}{\omega_{L,1}}\right), \label{R1}
\\
& R_2=\frac{T_{\mathrm{cm}}-\tau_1}{T_{\mathrm{fr}}} \log_{2}\left(1+\gamma_2\frac{P_{L,2}}{\omega_{L,2}}\right),\label{R2}
 \end{align}
 where $\gamma_i\equiv \frac{\lambda^2 d_{i}^{-\alpha}}{8 \pi N_0W_{\mathrm{tot}}}$
 is the SNR scaling factor, $\alpha$ is the path loss exponent, $N_0$ is the noise power spectral density, $W_{\mathrm{tot}}$ is the total bandwidth and $\omega_{L,i}\triangleq|\mathcal B_{L,i}|$ is the overall beam-width of the transmission beam.
 These equations presume that the transmission power $P_{L,i}$ is spread evenly across the transmit directions defined by the beam shape $\mathcal B_{L,i}$,
 so that the received SNR is $ \gamma_i P_{L,i}/\omega_{L,i}$.
 We then express the energy expenditure as a function of 
 the rate requirements as 
  \begin{align}
& E_{1}\triangleq\tau_1 P_{L,1}=
\omega_{L,1}\epsilon_1\left(\tau_1\right),
\\
& E_{2}\triangleq(T_{\mathrm{cm}}-\tau_1) P_{L,2}=
\omega_{L,2}
\epsilon_2(T_{\mathrm{cm}}-\tau_1),
 \end{align}
 where we have defined
 \begin{align}
 \label{epsilon}
 \epsilon_i(\tau)\triangleq
 \tau\frac{2^{\frac{T_{\mathrm{fr}}}{\tau}R_i}-1}{\gamma_i}
 \end{align}
 as the energy per radian required to transmit with average rate $R_i$ to MU$_i$ over an interval of duration $\tau$.

 \section{Optimization and Analysis}
 \label{analysis}
We define a policy $\pi$ as a function that, given the 
PDF $S_k$,
 selects 
the beam-shape $\mathcal B_k$ in each beam-alignment slot $k=0,1,\dots,L-1$,
the power $P_{L,i}$, beam-shape $\mathcal B_{L,i}$ and time allocation $\tau_1,T_{\mathrm{cm}}-\tau_1$ for both MUs during the 
data communication interval. The goal is to design $\pi$ so as to minimize the average power consumption in the data communication phase, with rate constraints $R_1$ and $R_2$ for both MUs. This optimization problem is expressed as
\begin{align}
\label{P1}
\bar P_{\mathrm{avg}}\triangleq\underset{\pi}{\text{min}}\ 
\mathbb{E}_\mu \Big[
\frac{\omega_{L,1}}{T_{\mathrm{fr}}}\epsilon_1\left(\tau_1\right)
+\frac{\omega_{L,2}}{T_{\mathrm{fr}}}\epsilon_2\left(T_{\mathrm{cm}}-\tau_1\right)\Big],
\end{align}
  where the expectation is with respect to the beam-shapes and time allocation
  prescribed by policy $\pi$, and the angular coordinates of the MUs.
We neglect the energy consumption in the beam-alignment phase,
studied in \cite{hussainenergy} for the single-user case,
 and thus assume that data communication is the most energy-hungry operation.
\vspace{-2mm}
\subsection{Markov Decision Process formulation}
We formulate the optimization problem as a MDP, with state
 given by the PDF of the angular coordinates of the two MUs, 
$S_k$ in slots $k=0,1,\dots,L$. 
During the beam-alignment phase, policy $\pi$ dictates the beam-shape in slot $k$ as
\begin{align}
\mathcal B_k=\pi_k(S_k).
\end{align}
At the end of the beam-alignment phase,
 the BS selects the time allocation $\tau_1$ for MU$_1$ and $T_{\mathrm{cm}}-\tau_1$ for MU$_2$
to be used during the communication phase. As explained previously, 
the beam-shape is chosen via (\ref{BS}) to provide coverage, and the power $P_{L,i}$
via \eqref{R1}-\eqref{R2} to support the rate demands.
Thus, policy $\pi$ dictates the time allocation as
\begin{align}
\tau_1=\pi_L(S_L).
\end{align}
Given the PDF $S_k$ and the beam-shape $\mathcal B_k$ during the beam-alignment slots,
the MUs generate the feedback $(C_{k,1},C_{k,2})$ via \eqref{feedback},
 with probability distribution
\begin{align}
\mathbb P\left(\left.C_{k,1}{=}c_1,C_{k,2}{=}c_2\right|S_k,\mathcal B_k\right)
=\!\!\!\!\!\!\!
\int \limits _{\mathcal B_k^{c_1}\times \mathcal B_k^{c_2}}\!\!\!\!\!S_k(\theta_1,\theta_2)\mathrm{d}\theta_1\mathrm d\theta_2,
\end{align}
where we have defined the set operation
\begin{align}
\label{setoper}
\mathcal A^1\equiv\mathcal A,\ \mathcal A^0\equiv [0,2\pi]\setminus \mathcal A.
\end{align}

Optimizating $\pi$ is challenging due to the continuous PDF space.
We now prove some structural properties of the model, which allow to simplify the state space.
\begin{theorem}
\label{theorem_1}
We have that
\begin{align}
\label{indep}
& S_k(\theta_1,\theta_2)=S_{k,1}(\theta_1)\cdot S_{k,2}(\theta_2) &\!\!\!\!\!\!\text{(independence)},
\\
& S_{k,i}(\theta_i)=\frac{1}{|\mathcal S_{k,i}|}\chi(\theta_i\in\mathcal S_{k,i}) &\!\!\!\!\!\!\text{(uniform distribution)}.
\label{unif}
\end{align}
Moreover, either $\mathcal S_{k,1}\equiv \mathcal S_{k,2}$
or $\mathcal S_{k,1}\cap\mathcal S_{k,2}\equiv \emptyset$.
\end{theorem}
\begin{proof}
See Appendix A.
\end{proof}

The independence and uniform distribution expressed by Theorem \ref{theorem_1} 
imply that the feedback signals generated by the two MUs are statistically independent of each other, \emph{i.e.},
$\mathbb P\left(\left.C_{k,1}=c_1,C_{k,2}=c_2\right|S_k,\mathcal B_k\right)
=\mathbb P\left(\left.C_{k,1}=c_1\right|S_{k,1},\mathcal B_k\right)
\mathbb P\left(\left.C_{k,1}=c_1\right|S_{k,2},\mathcal B_k\right),
$
with the probability of ACK given by
\begin{align}
\label{pack_iid}
\mathbb P\left(\left.C_{k,i}=1\right|S_{k,i},\mathcal B_k\right)
=\frac{|\mathcal B_k\cap\mathcal S_{k,i}|}{|\mathcal S_{k,i}|},
\end{align}
since $\Theta_i$ is uniformly distributed in the support $\mathcal S_{k,i}$.
The next state $S_{k+1}$ is then a deterministic function of the PDF $S_k$, beam-shape $\mathcal B_k$ and feedback $(c_{k,1},c_{k,2})$ via Bayes' rule, as in (\ref{Bayes}), and the support $\mathcal S_{k+1,i}$ for each MU is given by
\begin{align}
\label{update_rule}
\mathcal{S}_{k+1,i}\equiv\mathcal{S}_{k,i} \cap \mathcal{B}_{k}^{C_{k,i}},\ \forall i\in\{1,2\}.
\end{align}


We define the \emph{uncertainty width} for MU$_i$ as
 $U_{k,i}\triangleq |\mathcal S_{k,i}|$. Note that, the larger $U_{k,i}$, the more the uncertainty on the angular coordinate of MU$_i$. Additionally, let 
$\rho_k\triangleq\chi(\mathcal S_{k,1}\equiv\mathcal S_{k,2})$ be the binary variable indicating whether 
$\mathcal S_{k,1}\equiv\mathcal S_{k,2}$ (the two MUs are within the same uncertainty region, $\rho_k=1$)
or 
$\mathcal S_{k,1}\cap\mathcal S_{k,2}\equiv\emptyset$ (the two MUs are in different uncertainty regions, $\rho_k=0$).
We also define $\omega_i\triangleq |\mathcal S_{k,i}\cap \mathcal B_k|$ as the beam-width within the uncertainty region of MU$_i$.
Note that, if $\rho_k=1$, then it follows that $\mathcal S_{k,1}\equiv\mathcal S_{k,2}$, hence $\omega_{k,1}=\omega_{k,2}$.
We have the following result.
\begin{theorem}
\label{thm2}
$(U_{k,1},U_{k,2},\rho_k)$ is a sufficient statistic to select $(\omega_1,\omega_2)$ at time $k$.
Given $(\omega_1,\omega_2)$, the beam-shape $\mathcal B_k$ may be arbitrary provided that $|\mathcal S_{k,i}\cap\mathcal B_k|=\omega_i,\forall i\in\{1,2\}$.
\end{theorem}
\begin{proof}
See Appendix B.
\end{proof}

Therefore, in the following we can focus on the design of the beam-widths
$(\omega_1,\omega_2)$.
With this notation, the probability of ACK can be written as
\begin{align}
\label{PC}
\mathbb P\left(\left.C_{k,i}=1\right|U_{k,i},\omega_{k,i}\right)
=\frac{\omega_{k,i}}{U_{k,i}}.
\end{align}
Thus, given the state $(U_{k,1},U_{k,2},\rho_k)$ in slot $k=0,1,\dots,L{-}1$ and the beam-widths $(\omega_{k,1},\omega_{k,2})$,
 the new state becomes
 $(U_{k+1,1},U_{k+1,2},\rho_{k+1})$
 where
 \begin{align}
 \label{new_state_U}
U_{k+1,i}=
\left\{
\begin{array}{ll}
|\mathcal B_k\cap\mathcal S_{k,i}|=\omega_{k,i} & C_{k,i}=1,
\\
|\mathcal B_k^0\cap\mathcal S_{k,i}|=U_{k,i}-\omega_{k,i} & C_{k,i}=0,
\end{array}
\right.
 \end{align}
and
\begin{align}
\label{new_state_r}
\rho_{k+1}=
\left\{
\begin{array}{ll}
\rho_{k} & c_{k,1}=c_{k,2},
\\
0 & c_{k,1}\neq c_{k,2},
\end{array}
\right.
\end{align}
 with probabilities given by \eqref{PC}.
 The rule
\eqref{new_state_U} expresses the fact that, if an ACK is received, then 
the support of $S_{k+1,i}$ becomes
$\mathcal{S}_{k+1,i}\equiv\mathcal{S}_{k,i} \cap \mathcal{B}_{k}$ as given by \eqref{update_rule},
with width $\omega_{k,i}=|\mathcal{S}_{k+1,i}|$.
In contrast, if a NACK is received, then MU$_i$ is located in the complement region
$\mathcal{S}_{k+1,i}\equiv\mathcal{S}_{k,i}\setminus\mathcal{B}_{k}$, with width $U_{k,i}-\omega_{k,i}$. Rule
\eqref{new_state_r} describes the evolution of $\rho_k$. When $\rho_k=0$, the two MUs are located in disjoint uncertainty regions. In the next slot, they will still be in disjoint regions, irrespective of the feedback received at the BS.
In contrast, when $\rho_k=1$,  if the MUs send discordant feedback signals ($C_{k,1}\neq C_{k,2}$), the BS infers that they are located in disjoint
uncertainty regions, hence $\rho_{k+1}=0$; if the MUs send concordant feedback signals ($C_{k,1}=C_{k,2}$),
the BS infers that they are still in the same uncertainty region, hence $\rho_{k+1}=1$.
The optimal beam-alignment algorithm and MU scheduling can be found via dynamic programming (DP).
At the beginning of the communication phase,
given the state $(U_{L,1},U_{L,2},\rho_L)$
the optimal time allocation $\tau_1$ is the minimizer of (see (\ref{P1})
and (\ref{BS}) with $\omega_{L,i}=|\mathcal B_{L,i}|=|\mathcal S_{L,i}|=U_{L,i}$)
\begin{align}
\label{VL}
\!\!\!\!V_L\!(U_{L,1},\!U_{L,2},\!\rho_L)
\!=\!\!\!\!\!
\min_{\tau_1\in(0,T_{\mathrm{cm}})}\!\!\!\!
U_{L,1}\epsilon_1\!\left(\tau_1\right)
\!{+}U_{L,2}\epsilon_2\!\left(T_{\mathrm{cm}}\!{-}\tau_1\right)\!.\!\!
\end{align}
  Note that the objective function is convex in $\tau_1\in(0,T_{\mathrm{cm}})$, and it diverges for $\tau_1\to0$
  and $\tau_1\to T_{\mathrm{cm}}$.
  Thus, the optimal $\tau_1^*$ is the unique solver of
  \begin{align}
     \label{optimal_tau}
\frac{\epsilon_2^\prime\left(T_{\mathrm{cm}}-\tau_1^*\right)}
 {\epsilon_1^\prime\left(\tau_1^*\right)}=
 \frac{U_{L,1}}{U_{L,2}},
  \end{align}
  where $\epsilon_i^\prime(\tau)$
  is the first order derivative of $\epsilon_i(\tau)$ with respect to $\tau$.
  The function $V_L(U_{L,1},U_{L,2},\rho_L)$ denotes the cost-to-go function at the beginning of the communication phase. During the beam-alignment phase (slots $k=0,1,\dots,L-1$), the optimal value function for the cases $\rho_k=1$ and $\rho_k=0$  is computed recursively as
  \begin{align}
  \label{r_1}
&\!\!\!\!V_k(U,U,1)
{=}
\!\!\!\min_{\omega\in[0,U]}
\mathbb E\!\!\left[V_{k{+}1}(U_{k{+}1,1},U_{k{+}1,2},\rho_{k{+}1})\left|\!\!
\begin{array}{l}
\omega_{k,i}{=}\omega,\\
U_{k,i}{=}U,
\\ 
\rho_k{=}1
\end{array}
\right.\!\!\!
\right]
\nonumber \\&
\!\!\!\!=
\min_{\omega\in[0,U]}
\frac{\omega^2}{U^2}V_{k+1}(\omega,\omega,1)
{+}
\left(1{-}\frac{\omega}{U}\right)^2V_{k+1}(U{-}\omega,U{-}\omega,1)
\nonumber \\&
+\frac{\omega}{U}\left(1{-}\frac{\omega}{U}\right)
\bigr[V_{k+1}(\omega,U{-}\omega,0)+V_{k+1}(U{-}\omega,\omega,0)\bigr]
  \end{align}
and
\begin{align}
\label{r_0}
&\!\!\!\!\!V_k(U_1,U_2,0)
{=}\!\!\!\min_{\omega_i\in[0,U_i]}\!\!
\!\!\mathbb E\!\!\left[\!V_{k+1}(U_{k+1,1},U_{k+1,2},\rho_{k+1})\!\!\left|\!\!
\begin{array}{l}
\omega_{k,i}{=}\omega_i,\\
U_{k,i}{=}U_i,
\\ 
\rho_k{=}0
\end{array}
\right.\!\!\!\!\!
\right]
\nonumber \\&=
\nonumber
\min_{\omega\in[0,U_{k}]}
\frac{\omega_1}{U_1}\frac{\omega_2}{U_2}V_{k+1}(\omega_1,\omega_2,0)
\\&
+\frac{\omega_1}{U_1}\left(1{-}\frac{\omega_2}{U_2}\right)
V_{k+1}(\omega_1,U_2{-}\omega_2,0)\nonumber \\&
+\left(1{-}\frac{\omega_1}{U_1}\right)\frac{\omega_2}{U_2}
V_{k+1}(U_1{-}\omega_1,\omega_2,0)
\nonumber
\\&
+\left(1-\frac{\omega_1}{U_1}\right)
\left(1-\frac{\omega_2}{U_2}\right)
V_{k+1}(U_1-\omega_1,U_2-\omega_2,0).
\end{align}
These expressions are obtained by computing the expectation of
$V_{k+1}(U_{k+1,1},U_{k+1,2},\rho_{k+1})$, with respect to the realization of the feedback
signals $(C_{k,1},C_{k,2})$, with distribution \eqref{PC}, and the state dynamics given by \eqref{new_state_U} and \eqref{new_state_r}.
  
  In the next theorem, 
  we prove the optimality of a \emph{bisection} beam-alignment algorithm, which selects the
  beam-widths as $\omega_{k,i}=U_{k,i}/2$ in each slot.
 \begin{theorem}
 \label{thm3}
 The optimal beam-widths during the beam-alignment phase are given by
 \begin{align}
 \omega_{k,i}=\frac{1}{2}U_{k,i}.
 \end{align}
Then,
 \begin{align}
\bar P_{\mathrm{avg}}
=
\frac{\sigma}{T_{\mathrm{fr}}2^L}
\Big[\epsilon_1\left(\tau_1^*\right)+\epsilon_2\left(T_{\mathrm{cm}}-\tau_1^*\right)\Big],
\label{optvalue}
 \end{align}
 where $\tau_1^*$ uniquely satisfies
   \begin{align}
   \label{cbv}
\frac{\epsilon_2^\prime\left(T_{\mathrm{cm}}-\tau_1^*\right)}
 {\epsilon_1^\prime\left(\tau_1^*\right)}=1.
  \end{align}
\end{theorem}
\begin{proof}
See Appendix C.
\end{proof}

Note that, in the special case $\gamma_1=\gamma_2=\gamma$, \eqref{cbv} yields
\begin{align}
\label{tau_r1_r2}
\tau_1^*=\frac{R_1}{R_1+R_2}T_{\mathrm{cm}}
\end{align}
and 
 \begin{align}\label{speccase}
\bar P_{\mathrm{avg}}
=
\frac{\sigma}{\gamma 2^L}
\frac{T_{\mathrm{cm}}}{T_{\mathrm{fr}}}
\left(2^{\frac{T_{\mathrm{fr}}}{T_{\mathrm{cm}}}(R_1+R_2)}-1\right).
 \end{align}
 
\section{Numerical Results}
\label{numres}
 In this section, we compare the total power consumption
 versus the sum rate $R_{\mathrm{tot}}=R_1+R_2$ under:
 \begin{itemize}
 \item The proposed joint beam-alignment bisection algorithm.
  \item Single-user beam-alignment \cite{hussain2017throughput}: in this scheme, odd frames are allocated to MU$_1$ using $L_1$ slots for beam-alignment,
 even frames to MU$_2$ using $L_2$ slots for beam-alignment. 
 Beam-alignment is executed using the bisection scheme, whose optimality has been proved 
 in \cite{hussain2017throughput} for the single user scheme.
 To achieve the target rate demand $R_i$ over a period of two frames, 
 the rate demand for MU$_i$ is set to $2\times R_i$ in the corresponding allocated frame.
 \item Joint exhaustive search: 
the BS scans exhaustively
 up to $K=2^L$ beams, each with beam-width $2\pi/K$, starting from beam index $1$ to beam index $K$.
 When both MUs are detected, the communication phase starts, using the TDM scheme described in Section \ref{comm}. 
If MU$_i$ is located in the beam with index $\mathrm{id}_i$, beam-alignment will take
 $\max\{\mathrm{id}_1,\mathrm{id}_2\}$ slots, followed by data communication over 
 the remaining interval $T_{\mathrm{fr}}-\max\{\mathrm{id}_1,\mathrm{id}_2\}T$.
 \end{itemize}
 The parameters $L$, $L_1$, $L_2$ are optimized to achieve the minimum power consumption, constrained to $L,L_1,L_2\leq 7$. Thus, the minimum beam resolution is given by $2\pi/128$.
 
 We consider this scenario: $T_{\mathrm{fr}}=2\mathrm{ms}$, 
 $\sigma=2\pi$, $T=10\mu\mathrm{s}$, $d_{i}=50\mathrm{m}$, $W_{tot}=500\mathrm{MHz}$, $\lambda=5\mathrm{mm}$ (carrier frequency $60\mathrm{GHz}$), $\alpha=2$,
 $N_0=-174\mathrm{dBm}$. It follows that $\gamma_1=\gamma_2$.
 We vary $R_1$ and let $R_2=\psi R_1$, for a fixed parameter $\psi\in[0,1]$.
 
 
  
The results are plotted  in Fig. \ref{power_rate_1}. We notice that, when the rate for both MUs are equal ($\psi=1$), both joint and single-user beam-alignment have the same performance.
We note that
 the power consumption under the joint beam-alignment scheme with bisection 
 is independent of $\psi$, but only depends on the sum rate, as can be seen in  \eqref{speccase}. Using a similar argument as to derive \eqref{speccase}, the same holds under joint exhaustive search.
 In contrast, the power consumption under
 single-user beam-alignment is highly affected by $\psi$. This is due to the fact that 
 an entire frame is allocated to MU$_2$, despite its rate demand is only a fraction $\psi$ of that of MU$_1$. This causes great imbalances in the power allocated to the two MUs (such imbalance disappears when $\psi=1$, so that the rate demands are the same). 
Instead, with joint beam-alignment, the two MUs are scheduled optimally based on TDM, yielding significant power savings.
We note that the joint beam-alignment scheme with bisection has the least power consumption, with $3\mathrm{dB}$ power saving compared to joint exhaustive search,
and up to $7\mathrm{dB}$ compared to single-user beam-alignment, for moderate imbalances on the rate demands ($\psi=0.5$).

\begin{figure}
\begin{center}
\includegraphics[width=0.9\columnwidth]{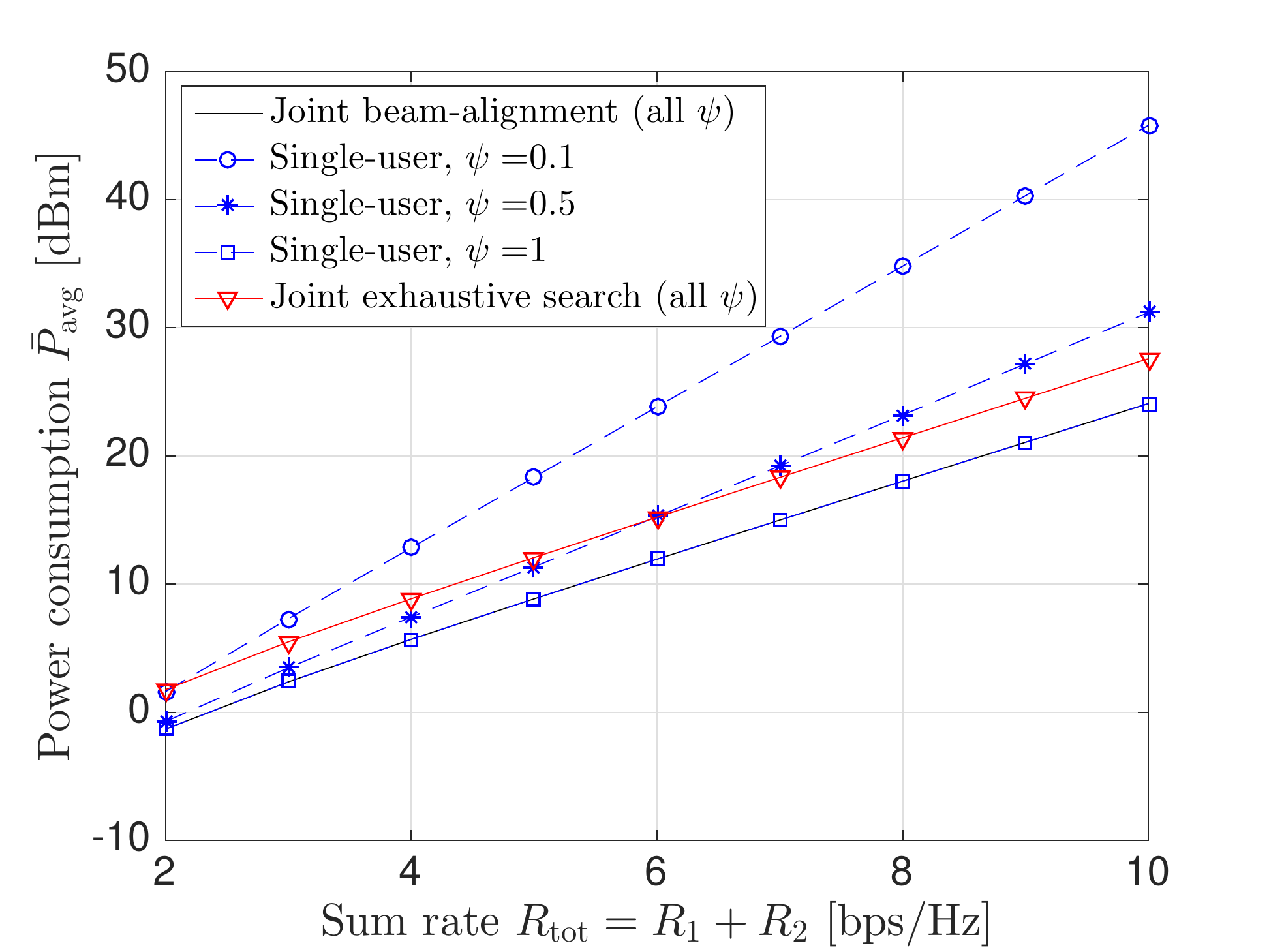}
\caption{Power versus sum rate under different algorithms.
}
\label{power_rate_1}
\end{center}
\vspace{-7mm}
\end{figure} 
\vspace{-2mm}
\section{Conclusions}
\label{conclu}
In this paper, we studied the design of energy efficient joint beam-alignment protocols for two users, with the goal to minimize the power consumption during data transmission, subject to rate constraints for both users, under analog beamforming constraints. We prove that a bisection search algorithm is optimal. In addition we  schedule
optimally the two users during data communication via time division multiplexing, based on the outcome of beam-alignment.
Our numerical results show significant power savings compared to exhaustive search and a single-user beam-alignment scheme taking place separately for each user.

\vspace{-4mm}
\section*{Appendix A: Proof of Theorem \ref{theorem_1}}
Note that \eqref{feedback} along with Bayes' rule \eqref{Bayes} imply \eqref{update_rule}.
We prove the theorem by induction.
The induction hypothesis holds for $k=0$, see \eqref{init}.
Now, assume that it holds in slot $k\geq 0$. We show that 
this implies that it holds in slot $k+1$ as well.
Thus, assume that either $\mathcal S_{k,1}\equiv \mathcal S_{k,2}$
or $\mathcal S_{k,1}\cap\mathcal S_{k,2}\equiv \emptyset$.
First, let us consider the case $\mathcal S_{k,1}\equiv \mathcal S_{k,2}$ with $C_{k,1}=C_{k,2}$.
From \eqref{update_rule} we have that 
\begin{align*}
\mathcal{S}_{k+1,1}\equiv\mathcal{S}_{k,1} \cap \mathcal{B}_{k}^{C_{k,1}}
\equiv\mathcal{S}_{k,2} \cap \mathcal{B}_{k}^{C_{k,2}}\equiv\mathcal{S}_{k+1,2},
\end{align*}
and thus  $\mathcal{S}_{k+1,1}\equiv\mathcal{S}_{k+1,2}$.
For all other cases, we have that 
\begin{align*}
\mathcal{S}_{k+1,1}\cap\mathcal{S}_{k+1,2}
\equiv(\mathcal{S}_{k,1}\cap\mathcal{S}_{k,2}) \cap (\mathcal{B}_{k}^{C_{k,1}} \cap \mathcal{B}_{k}^{C_{k,2}})
\equiv\emptyset,
\end{align*}
since either  $\mathcal{S}_{k,1}\cap\mathcal{S}_{k,2}\equiv\emptyset$ from the induction hypothesis,
or $\mathcal{S}_{k,1}\equiv\mathcal{S}_{k,2}$ but $C_{k,1}\neq C_{k,2}$, yielding
 $\mathcal{B}_{k}^{C_{k,1}}\cap\mathcal{B}_{k}^{C_{k,2}}\equiv\mathcal{B}_{k}^{0}\cap\mathcal{B}_{k}^{1}\equiv\emptyset$.
 Thus, it follows that either $\mathcal S_{k,1}\equiv \mathcal S_{k,2}$
or $\mathcal S_{k,1}\cap\mathcal S_{k,2}\equiv \emptyset$.

Now, assume that 
$S_k$ satisfies \eqref{indep} and \eqref{unif} in slot $k$.
By specializing Bayes' rule \eqref{Bayes} to this case, we obtain
\begin{align}
&S_{k+1}(\theta_1,\theta_2)=
\frac{
\prod_{i\in\{1,2\}}\chi(\theta_i\in\mathcal S_{k,i}\cap\mathcal B_k^{c_{k,i}})
}{\int _{-\pi}^{\pi}\int_{-\pi}^{\pi}
\prod_{i\in\{1,2\}}\chi(\tilde\theta_i\in\mathcal S_{k,i}\cap\mathcal B_k^{c_{k,i}})
\mathrm d\tilde\theta_{1}\mathrm d\tilde\theta_{2}},
\nonumber
\end{align}
where we have used \eqref{feedback} and \eqref{setoper},
and the fact that 
$\chi(\theta_i\in\mathcal B_k^{c_{k,i}})\chi(\theta_i\in\mathcal S_{k,i})
=\chi(\theta_i\in\mathcal S_{k,i}\cap\mathcal B_k^{c_{k,i}})$. Solving the integral in 
the denominator and using \eqref{update_rule} we obtain
\begin{align}
&S_{k+1}(\theta_1,\theta_2)=
\prod_{i\in\{1,2\}}
\frac{1}{|\mathcal S_{k+1,i}|}\chi(\theta_i\in\mathcal S_{k+1,i}),
\nonumber
\end{align}
thus proving the induction step. The theorem is proved.

\section*{Appendix B: Proof of Theorem \ref{thm2}}
We prove this theorem by induction.
Let $V_k(S_k)$ be the value function from state $S_k$ in slot $k$.
 At the beginning of the communication phase, from \eqref{P1} we have that
 \begin{align}
V_L(S_L)=\underset{\tau_1}{\text{min}}\ 
U_{L,1}\epsilon_1\left(\tau_1\right)
+U_{L,2}\epsilon_2\left(T_{\mathrm{cm}}-\tau_1\right),
\end{align}
since the condition \eqref{BS} implies
$\omega_{L,i}=|\mathcal B_{L,i}|=|\mathcal S_{L,i}|=U_{L,i}$.
Therefore,
 \begin{align}
V_L(S_L)=V_L(|S_{L,1}|,|S_{L,2}|,\chi(S_{L,1}\equiv S_{L,2})).
 \end{align}
Now, let $k<L$ and assume that $(U_{k+1,1},U_{k+1,2},\rho_{k+1})$
is a sufficient statistic to choose $\omega_{j,i}$ for $j\geq k+1$, with $\mathcal B_j$
such that $|\mathcal S_{j,i}\cap\mathcal B_j|=\omega_{j,i},\forall i\in\{1,2\}$, \emph{i.e.},
 \begin{align*}
V_{k+1}(S_{k+1})=V_{k+1}(|S_{k+1,1}|,|S_{k+1,2}|,\chi(S_{k+1,1}\equiv S_{k+1,2})).
 \end{align*}
The dynamic programming iteration yields
\begin{align}
&V_k(S_k)=
\min_{\mathcal B_k}\mathbb E\left[V_{k+1}(S_{k+1})|\mathcal B_k,\mathcal S_k\right]
\\&
\!{=}
\min_{\mathcal B_k}\mathbb E\left[
V_{k+1}(|S_{k+1,1}|,|S_{k+1,2}|,\chi(S_{k+1,1}{\equiv}S_{k+1,2}))
|\mathcal B_k,\mathcal S_k\right]\!\!,
\nonumber
 \end{align}
 where we have used the induction hypothesis.
 
 Note that $\mathcal{S}_{k+1,i}$ is obtained via \eqref{update_rule}.
 If $\mathcal{S}_{k,1}\equiv\mathcal{S}_{k,2}$ ($\rho_k=1$), it follows that
 $\mathcal{S}_{k+1,1}\equiv\mathcal{S}_{k+1,2}$  ($\rho_{k+1}=1$) iff
 $C_{k,1}=C_{k,2}$, yielding $\rho_{k+1}=\chi(C_{k,1}=C_{k,2})$.
 If  $\mathcal{S}_{k,1}\cap\mathcal{S}_{k,2}\equiv\emptyset$ ($\rho_k=1$), it follows that
  $\mathcal{S}_{k+1,1}\cap\mathcal{S}_{k+1,2}\equiv\emptyset$, hence $\rho_{k+1}=0$.
  Therefore, we can write
  \begin{align}
  \rho_{k+1}=\rho_k\chi(C_{k,1}=C_{k,2}).
  \end{align}

 By computing the expectation with respect to the feedback distribution given by \eqref{pack_iid},
 and using \eqref{update_rule}, we then obtain
 \begin{align*}
&V_k(S_k)
=
\min_{\mathcal B_k}
\sum_{(c_1,c_2)\in\{0,1\}^2}
\frac{|\mathcal B_k^{c_1}\cap\mathcal S_{k,1}|}{|\mathcal S_{k,1}|}
\frac{|\mathcal B_k^{c_2}\cap\mathcal S_{k,2}|}{|\mathcal S_{k,2}|}
\\&
\times
V_{k+1}(|\mathcal{S}_{k,1} \cap \mathcal{B}_{k}^{c_1}
|,|\mathcal{S}_{k,2} \cap \mathcal{B}_{k}^{c_2}|,\rho_k\chi(c_1=c_2)
)).
 \end{align*}
 Now, letting $\omega_{k,i}\triangleq|\mathcal{S}_{k,i}\cap \mathcal{B}_{k}|$
 and $|\mathcal{S}_{k,i}|=U_{k,i}$, we find that
 $|\mathcal{S}_{k,i}\cap \mathcal{B}_{k}^0|=U_{k,i}-\omega_{k,i}$, yielding
  \begin{align*}
&V_k(S_k)
=
\min_{\mathcal B_k}
\sum_{(c_1,c_2)\in\{0,1\}^2}
\prod_{i\in\{1,2\}}\frac{\omega_{k,i}^{c_i}(U_{k,i}-\omega_{k,i})^{1-c_i}}{U_{k,i}}\
\\&
\times
V_{k+1}(
\omega_{k,1}^{c_1}(U_{k,1}-\omega_{k,1})^{1-c_1}
,
\omega_{k,2}^{c_2}(U_{k,2}-\omega_{k,2})^{1-c_2}
\\&\qquad\qquad\qquad\qquad\qquad\qquad\qquad\qquad
,\rho_k\chi(c_1=c_2)
).
 \end{align*}
 Note that, given $(\omega_{k,1},\omega_{k,2})$
 and $(U_{k,1},U_{k,2},\rho_k)$,
 the objective function  is independent of $\mathcal B_k$ and $S_k$.
 Thus, the minimization over $\mathcal B_k$ can be restricted to a minimization over 
 $(\omega_{k,1},\omega_{k,2})$, with the additional constraint that
 $\omega_{k,1}=\omega_{k,2}$ if $\mathcal S_{k,1}\equiv\mathcal S_{k,2}$ ($\rho_{k}=0$),
 yielding
\begin{align}
V_k(S_k)=V_k(U_{k,1},U_{k,2},\rho_k).
\end{align}
The induction step is proved, hence the theorem.

  \section*{Appendix C: Proof of Theorem \ref{thm3}}
  We prove the theorem by induction. In particular, we show that, for all $k=0,1,\dots, L$,
\begin{align}
\label{dfg}
V_k(U_1,U_2,\rho)=V_L\left(\frac{U_1}{2^{L-k}},\frac{U_2}{2^{L-k}},0\right),\ \forall\rho\in\{0,1\}.
\end{align}
This condition clearly holds for $k=L$, since
$V_L(U_1,U_2,\rho)=V_L\left(U_1,U_2,0\right)$ from \eqref{r_0}. Thus, let $k<L$ and assume that 
\begin{align}
V_{k+1}(U_1,U_2,\rho)=V_L\left(\frac{U_1}{2^{L-k-1}},\frac{U_2}{2^{L-k-1}},0\right).
\end{align}
We prove that this implies (\ref{dfg}).
$V_{k}$ is computed from $V_{k+1}$ via DP, as in (\ref{r_1}) and \eqref{r_0}.

We start from the case $\rho_k=0$, and then consider the case $\rho_k=1$ (which implies $U_{k,1}=U_{k,2}$ 
and $\omega_{k,1}=\omega_{k,2}$).
Let
\begin{align}
g(x_1,x_2)\triangleq x_1x_2V_{k+1}(x_1,x_2,0).
\end{align}
Then, we can write the DP recursion \eqref{r_0} as
\begin{align}
\label{recursion_DP}
& V_{k}(U_1,U_2,0)
=\min_{\omega_i\in[0,U_i]}
\frac{
\frac{1}{2}g(\omega_1,\omega_2)
+\frac{1}{2}g(U_1-\omega_1,\omega_2)
}{U_1U_2/2}
\nonumber \\&
+\frac{
\frac{1}{2}g(\omega_1,U_2-\omega_2)
+\frac{1}{2}g(U_1-\omega_1,U_2-\omega_2)
}{U_1U_2/2}.
\end{align}
We denote the objective function in \eqref{recursion_DP} as $h(\omega_1,\omega_2)$, so that
we can rewrite $V_{k}(U_1,U_2,0)=\min_{\omega_i\in[0,U_i]}h(\omega_1,\omega_2)$.
In the final part of the proof, we will show that $g(x_1,x_2)$ is a convex function of $x_i,i\in\{1,2\}$ (although not necessarily jointly convex with respect to $(x_1,x_2)$).
By applying Jensen's inequality to $h(\omega_1,\omega_2)$ in \eqref{recursion_DP},
first with respect to the first argument of the function $g(\cdot,\cdot)$, and then with respect to the second argument,
it follows that 
\begin{align*}
&h(\omega_1,\omega_2)
\geq 
\frac{
\frac{1}{2}g\left(\frac{U_1}{2},\omega_2\right)
+\frac{1}{2}g\left(\frac{U_1}{2},U_2-\omega_2\right)
}{U_1U_2/4}
\geq
\frac{
g\left(\frac{U_1}{2},\frac{U_2}{2}\right)
}{U_1U_2/4}.
\end{align*}

Thus, it follows that 
\begin{align}
& V_{k}(U_1,U_2,0)
\geq 
4\frac{
g(U_1/2,U_2/2)
}{U_1U_2}.
\end{align}
Indeed, it can be seen by inspection that such lower bound is achievable by the bisection policy 
$\omega_i=U_i/2$, which proves the induction step for the case $\rho_k=0$.

We now consider the case $\rho_k=1$. Using the fact that $V_{k+1}(U_1,U_2,0)=V_{k+1}(U_1,U_2,1)$
from the  induction hypothesis, from \eqref{r_1} we obtain
  \begin{align}
  &V_k(U,U,1)
{=}
\!\!\!\min_{\omega\in[0,U]}
\mathbb E\!\!\left[V_{k{+}1}(U_{k{+}1,1},U_{k{+}1,2},0)\!\!\left|\!\!
\begin{array}{l}
\omega_{k,i}{=}\omega,\\
U_{k,i}{=}U,
\\ 
\rho_k{=}1
\end{array}
\right.\!\!\!
\right]
\nonumber \\&
\geq
\min_{(\omega_1,\omega_2)\in[0,U]^2}
\mathbb E\!\!\left[V_{k{+}1}(U_{k{+}1,1},U_{k{+}1,2},0)\!\!\left|\!\!
\begin{array}{l}
\omega_{k,i}{=}\omega_i,\\
U_{k,i}{=}U,
\\ 
\rho_k{=}1
\end{array}
\right.\!\!\!
\right]
\nonumber
\\&
=
V_k(U,U,0),
  \end{align}
  where the inequality follows from the fact that we have extended the optimization interval to
  $(\omega_1,\omega_2)\in[0,U]^2$,
 and therefore $V_k(U,U,1)\geq V_k(U,U,0)$.
We have seen that, for the case $\rho_k=0$, the value function is optimized by the bisection policy.
By inspection, we can see that the lower bound $V_k(U,U,0)$ is also attained by the bisection policy $\omega_{k,1}=\omega_{k,2}=U/2$, which 
satisfies the requirement 
$\omega_{k,1}=\omega_{k,2}$ when $\rho_k=1$. Thus, we have proved the induction step.

By letting $k{=}0$ in \eqref{dfg} with $U_1{=}U_2{=}\sigma$,
 and using \eqref{VL}, we finally obtain \eqref{optvalue} after dividing the energy consumption by the frame duration $T_{\mathrm{fr}}$.
  $\tau_1^*$ is the unique solution of 
 \eqref{optimal_tau}, yielding \eqref{cbv} since $U_{L,i}=\sigma/2^L$ under bisection.

It remains to prove that $g(x_1,x_2)$
is a convex function of $x_i,i\in\{1,2\}$. Due to the symmetry of $g(x_1,x_2)$ with respect to its arguments, it is sufficient to prove convexity with respect to $x_1$ only, with $x_2$ fixed.
We have
\begin{align}
\nonumber
&g(x_1,x_2)=
x_1x_2
V_L\left(\frac{x_1}{2^{L-k-1}},\frac{x_2}{2^{L-k-1}},0\right)
\\&
=
\frac{1}{2^{L-k-1}}
\min_{\tau_1\in(0,T_{\mathrm{cm}})}
x_1^2x_2\epsilon_1\left(\tau_1\right)
+x_1x_2^2\epsilon_2\left(T_{\mathrm{cm}}-\tau_1\right).
\nonumber
\end{align}
Note that the convexity of $g(\cdot)$ is unaffected by $k$, thus we let $k=L-1$.
Let $\tau_1(x_1)$ be the minimizer above, as a function of $x_1$.
We obtain
\begin{align*}
&\frac{\mathrm dg(x_1,x_2)}{\mathrm dx_1}= 
2x_1x_2\epsilon_1\left(\tau_1(x_1)\right)
+x_2^2\epsilon_2\left(T_{\mathrm{cm}}-\tau_1(x_1)\right)
\nonumber \\&+\tau_1^\prime(x_1)x_1x_2
\left[
x_1\epsilon_1^\prime\left(\tau_1(x_1)\right)
-x_2\epsilon_2^\prime\left(T_{\mathrm{cm}}-\tau_1(x_1)\right)
\right],
\end{align*}
where $\tau_1^\prime(x_1)\triangleq \frac{\mathrm d\tau_1(x_1)}{\mathrm dx_1}$.
Note that $\tau_1(x_1)$ must satisfy \eqref{optimal_tau} (with $U_{L,i}=x_i$), yielding
\begin{align*}
\frac{\mathrm dg(x_1,x_2)}{\mathrm dx_1}=
2x_1x_2\epsilon_1\left(\tau_1(x_1)\right)
+x_2^2\epsilon_2\left(T_{\mathrm{cm}}-\tau_1(x_1)\right).
\end{align*}
The second derivative of $g(x_1,x_2)$ with respect to $x_1$ is then given by
\begin{align}
&\frac{\mathrm d^2g(x_1,x_2)}{\mathrm dx_1^2}=
2x_2\epsilon_1\left(\tau_1(x_1)\right)
+x_1x_2\epsilon_1^\prime\left(\tau_1(x_1)\right)\tau_1^\prime(x_1)
\nonumber \\ &+\tau_1^\prime(x_1)x_2
\left[
x_1\epsilon_1^\prime\left(\tau_1(x_1)\right)
-x_2\epsilon_2^\prime\left(T_{\mathrm{cm}}-\tau_1(x_1)\right)\right].
\end{align}
Using again the fact that $\tau_1(x_1)$ must satisfy \eqref{optimal_tau}, we obtain
\begin{align}
\label{2nd_derivative_x1}
\!\!\!\frac{\mathrm d^2g(x_1,x_2)}{\mathrm dx_1^2}=
2x_2\epsilon_1\left(\tau_1(x_1)\right)
{+}\epsilon_1^\prime\left(\tau_1(x_1)\right)x_1
x_2\tau_1^\prime(x_1).
\end{align}

From \eqref{optimal_tau}, we have that $\tau_1(x_1)$ must satisfy
$x_2\epsilon_2^\prime\left(T_{\mathrm{cm}}{-}\tau_1(x_1)\right)
{=}
x_1\epsilon_1^\prime\left(\tau_1(x_1)\right)$. By computing the derivative with respect to $x_1$ on both sides of this equation, we obtain
$\tau_1^\prime(x_1)$ as 
\begin{align}
\tau_1^\prime(x_1)
=
\frac{1}{x_2}
\frac{[\epsilon_1^\prime\left(\tau_1(x_1)\right)]^2
}{
\left[\begin{array}{l}
-\epsilon_1^\prime\left(\tau_1(x_1)\right)
\epsilon_2^{\prime\prime}\left(T_{\mathrm{cm}}-\tau_1(x_1)\right)
\\
-\epsilon_1^{\prime\prime}\left(\tau_1(x_1)\right)
\epsilon_2^\prime\left(T_{\mathrm{cm}}-\tau_1(x_1)\right)
\end{array}\right]
}.
\end{align}
Thus, by substituting in \eqref{2nd_derivative_x1},
 the convexity of $g(x_1,x_2)$ 
 ($\frac{\mathrm d^2g(x_1,x_2)}{\mathrm dx_1^2}>0$) becomes equivalent to
 \begin{align}
 \label{convexity_check}
-2\epsilon_1\epsilon_1^\prime\epsilon_2^{\prime\prime}
-2\epsilon_1\epsilon_1^{\prime\prime}\epsilon_2^\prime
+\epsilon_2^\prime[\epsilon_1^\prime]^2
>
0
\end{align}
where $\epsilon_i$, $\epsilon_i^\prime$, $\epsilon_i^{\prime\prime}$
is shorthand notation for
$\epsilon_i(\tau_i(x_1))$, $\epsilon_i^\prime(\tau_i(x_1))$, $\epsilon_i^{\prime\prime}(\tau_i(x_1))$, with
$\tau_2(x_1)=T_{\mathrm{cm}}-\tau_1(x_1)$, respectively.

Let  $y_1=\frac{T_{\mathrm{fr}}}{\tau_1}R_1$
and
$y_2=\frac{T_{\mathrm{fr}}}{T_{\mathrm{cm}}-\tau_1}R_2$.
We obtain
 \begin{align}
 \left\{
 \begin{array}{l}
\epsilon_i=
\frac{2^{y_i}-1}{y_i}\frac{T_{\mathrm{fr}}R_i}{\gamma_i},
\\
 \epsilon_i^\prime=
 \frac{2^{y_i}-1}{\gamma_i}-\frac{2^{y_i}}{\gamma_i}\ln(2)y_i,
\\
 \epsilon_i^{\prime\prime}=
 \frac{2^{y_i}}{\gamma_i T_{\mathrm{fr}}R_i}[\ln(2)]^2y_i^3.
 \end{array}\right.
 \end{align}
 Substituting in \eqref{convexity_check}, 
convexity becomes equivalent to
\begin{align}
\nonumber
&q(y_1,y_2)
\triangleq
2^{y_2}[\ln(2)y_2-1+2^{-y_2}]2[1-2^{-y_1}][\ln(2)]^2y_1^2\\&
-2^{y_2}[\ln(2)y_2-1+2^{-y_2}][\ln(2)y_1-1+2^{-y_1}]^2\\&
+2\frac{R_1}{R_2}[1-2^{-y_1}]
[\ln(2)y_1-1+2^{-y_1}][\ln(2)]^2\frac{2^{y_2}y_2^3}{y_1}
>0,
\nonumber
\end{align}
which we are now going to prove.
Using the fact that $\ln(2)y_1-1+2^{-y_1}>0$,
we have that 
\begin{align}
\nonumber
&q(y_1,y_2)
\geq
2^{y_2}[\ln(2)y_2-1+2^{-y_2}]2[1-2^{-y_1}][\ln(2)]^2y_1^2\\&
-2^{y_2}[\ln(2)y_2-1+2^{-y_2}][\ln(2)y_1-1+2^{-y_1}]^2
 \\ & \propto
 \nonumber
 2[1-2^{-y_1}][\ln(2)]^2y_1^2
 -[\ln(2)y_1-1+2^{-y_1}]^2
\triangleq \hat q(y_1),
\end{align}
where $\propto$ denotes proportionality up to the multiplicative positive factor $2^{y_2}[\ln(2)y_2-1+2^{-y_2}]>0$.
The derivative of $\hat q(y_1)$ with respect to $y_1$ is given by
\begin{align}
\nonumber
&\frac{\mathrm d\hat q(y_1)}{\mathrm dy_1}
=
2\ln(2)[1-2^{-y_1}]^2
\\&
+2\ln(2)\ln(2)y_1[1-2^{-y_1}+2^{-y_1}\ln(2)y_1]
>0.
\end{align}
Therefore, we obtain
$q(y_1,y_2)\geq 2^{y_2}[\ln(2)y_2-1+2^{-y_2}]\hat q(y_1)>2^{y_2}[\ln(2)y_2-1+2^{-y_2}]\hat q(0)=0$. 
The convexity of $g(x_1,x_2)$ with respect to $x_i$ is proved, hence the theorem.

  \vspace{-4mm}
\bibliographystyle{IEEEtran}
\bibliography{reference}

\end{document}